\documentclass[preprint,12pt]{article}

\usepackage{amssymb}

\usepackage{amsmath,latexsym,eufrak,euscript}
\usepackage{subfigure,pstricks,pst-node,pst-coil}

\usepackage{url,tikz}
\usepackage{pgflibrarysnakes}

\usepackage{multicol}

\usetikzlibrary{arrows}
\usetikzlibrary{automata}
\usetikzlibrary{snakes}
\usetikzlibrary{shapes}

\def \sp {\hspace*{0.8cm}}

\def \A {\mathcal{A}}

\def \C {\mathfrak{C}}
\def \N {\mathfrak{N}}

\def \P {\mathbb{P}}

\def \X {\mathfrak{X}}
\def \NFA {\mathfrak{A}}

\def \Aut {\mathrm{Aut}}
\def \ChI {\mathsf{Ch_{init}}}
\def \ChF {\mathsf{Ch_{final}}}
\def \ChT {\mathsf{Ch_{trans.}}}
\def \S {S_{\rho_1,\rho_2,\rho_3}}
\def \Sr {S_{\rho}}

\def \endofproof {\hfill $\square$}
\newcommand {\tv}[1] {\lVert #1 \rVert_{\text{\em {\tiny TV}}}}

\newenvironment{proof}
{ \noindent {\sc Proof.\/}  }
{}

\newtheorem{theorem}{Theorem}
\newtheorem{lemma}[theorem]{Lemma}
\newtheorem{proposition}[theorem]{Proposition}

\title{On the Uniform Random Generation of Non deterministic
  Automata up to Isomorphism}
\date{}
\author{Pierre-Cyrille H\'eam and Jean-Luc Joly}
\begin{document}


\maketitle

\begin{center}
FEMTO-ST  - Universit\'e de Franche-Comt\'e -  CNRS UMR 6174 - INRIA CASSIS
\end{center}
\begin{abstract}
In this paper we address the problem of the uniform random generation of non
deterministic automata (NFA) up to isomorphism. First, we show how to use a
Monte-Carlo approach to uniformly sample a NFA. 
Secondly, we show how to use the Metropolis-Hastings Algorithm to  uniformly
generate  NFAs up to isomorphism. Using
labeling techniques, we show that in practice it is possible to move into the
modified Markov Chain efficiently, allowing the random generation of NFAs up
to isomorphism with dozens of states. This general approach is also applied
to several interesting subclasses of NFAs (up to isomorphism), such as NFAs
having a unique initial states and a bounded output degree. Finally, we
prove that for these interesting subclasses of NFAs, moving into the
Metropolis Markov chain can be done in polynomial time. Promising
experimental results constitute a practical  contribution.
\end{abstract}


\section{Introduction}\label{sec:intro}

Finite automata play a central role in the field of formal language
theory and are intensively used to address algorithmic problems from
model-checking to text processing. Many automata-based algorithms have been
developed and are still  being developed. They propose new approaches and heuristics,
even for basic problems like the inclusion
problem\footnote{see \url{http://www.languageinclusion.org/}}.
Evaluating new algorithms is a challenging problem that cannot be addressed
only by the theoretical computation of the worst case complexity. Several
other complementary techniques can be used to measure the efficiency of an
algorithm: average complexity, generic case complexity, benchmarking,
evaluation on hard instances, evaluations on random instances. The first two
approaches are hard theoretical problems, particularly for algorithms
using heuristics and optmizations. Benchmarks, as well as known hard instances, are not always available. 
Nevertheless, in practice, random generation
of inputs is a good way to estimate the efficiency of an algorithm.
Designing uniform random generator for classes of finite automata is a
challenging problem that has been addressed mostly for deterministic
automata~\cite{DBLP:journals/tcs/ChamparnaudP05,DBLP:journals/tcs/BassinoN07,DBLP:journals/tcs/AlmeidaMR07,DBLP:conf/stacs/CarayolN12}
 -the interested reader is referred to~\cite{DBLP:conf/mfcs/Nicaud14} for a
recent survey. However, the problem of uniform random generation of non
deterministic automata (NFAs) is more complex, particularly for a random generation
up to isomorphism: the size of the automorphism group of a $n$-state non
deterministic automata may vary from $1$ to $n!$. 
For most applications, the complexity of the algorithm is related to the
structure of the automata, not to the names of the states:
randomly generated NFAs, regardless of the number of isomorphic automata, may therefore 
lead to an over representation of some isomorphism classes of automata.
Moreover, as discussed in the conclusion of~\cite{DBLP:conf/mfcs/Nicaud14},
the random generation of non deterministic automata has to be done on particular
subclasses of automata in order to obtain a better sampler for the
evaluation of algorithm (since most of the NFAs, for the uniform distribution,
will accept all words).  

The random generation of non deterministic automata is explored
in~\cite{DBLP:conf/lpar/TabakovV05} using random graph techniques (without
considering the obtained distribution relative to automata or to the
isomorphism classes). In~\cite{DBLP:conf/dcfs/ChamparnaudHPZ02}, the random
generation of NFAs is performed using bitstream generation.
In~\cite{DBLP:conf/lata/Nicaud09,DBLP:conf/fsttcs/NicaudPR10} NFAs are
obtained by the random generation of a regular expression and by
transforming it into an equivalent automaton using Glushkov Algorithm. The
use of Markov chains based techniques to randomly generate finite automata
was introduced
in~\cite{DBLP:conf/wia/CarninoF11,DBLP:journals/tcs/CarninoF12} for acyclic
automata.

\subsection{Contributions}
In this paper we address the problem of the uniform generation of elements
in several classes of non deterministic automata (up to isomorphism) by
using Monte-Carlo techniques. We propose this approach for the class of
$n$-state non deterministic automata as well as for (a priori) more
interesting sub-classes. Determining the most interesting subclasses of NFAs
for testing practical applications is not the purpose of this paper. We
would like to point out that Monte Carlo approaches are very flexible and
that the results of this paper can be applied-adapted quite easily for many
classes of NFAs. More precisely:
\begin{enumerate}
\item We propose in Section~\ref{sec:random} several ergodic Markov Chains
whose stationary distributions are respectively uniform on the set of
$n$-state NFAs, $n$-state NFAs with a fixed maximal output degree and
$n$-state NFAs with a fixed maximal output degree for each letter. In
addition, these chains can be adapted for these three classes restricted to
automata with a fixed single initial state. Moving into these Markov chains
can be done in time polynomial in $n$.
\item The main idea of this paper is exposed in
Section~\ref{sec:algo-metro}, where we show how to modify these Markov
Chains using the Metropolis-Hastings Algorithm in order to obtain stationary
distributions that are uniform for the given classes of automata up to
isomorphism. Moving into these new Markov chains requires to compute the sizes
of the automorphism groups of the occurring NFAs, as explained in Section~\ref{sec:randomup}. 

\item The main contributions of this paper are given in
Section~\ref{sec:poly} and in Section~\ref{sec:label}, which can be
red independently. In Section~\ref{sec:poly}, we show a theoretical result for
the classes with a bounded output degree: moving into the modified Markov chains
(for a generation up to isomorphism) can be done in polynomial time. 

\item 
In Section~\ref{sec:label}, we explain how to use {\it labelings}, a classical
graph technique, to compute in practice the sizes of the automorphism
groups. The efficiency of the approach is illustrated with promising
experiments in Section~\ref{sec:xp}.
\end{enumerate}



\subsection{Theoretical Background on NFA}\label{sec:theoryNFA}

For a general reference on finite automata see~\cite{Hopcroft}.  In
this paper $\Sigma$ is a fixed finite alphabet of cardinal
$|\Sigma|\geq 2$, and $m$ is an integer satisfying $m\geq 2$.

A {\it non-deterministic automaton} (NFA) on $\Sigma$ is a tuple $(Q,
\Delta,I,F)$ where $Q$ is a finite set of {\it states}, $\Sigma$ is a finite
alphabet, $\Delta\subseteq Q\times \Sigma \times Q$ is the set of
transitions, $F\subseteq Q$ is the set of final states and $I\subseteq
Q$ is the set of initial states.  For any state $p$ and any letter
$a$, we denote by $p\cdot a$ the set of states $q$ such that
$(p,a,q)\in\Delta$.  The set of transitions $\Delta$ is {\it
deterministic} if for every pair $(p,a)$ in $Q\times\Sigma$ there is
at most one $q\in Q$ such that $(p,a,q)\in \Delta$. Two NFAs are
depicted on Fig.~\ref{fig:isom}.  A NFA is {\it complete} if for every
pair $(p,a)$ in $Q\times\Sigma$ there is at least one $q\in Q$ such
that $(p,a,q)\in \Delta$.  A {\it path} in a NFA is a sequence of
transitions $(p_0,a_0,q_0)(p_1,a_1,q_1)\ldots(p_k,a_k,q_k)$ such that
$q_i=p_{i+1}$. The word $a_0\ldots a_k$ is the {\it label} of the path
and $k$ its {\it length}. If $p_0\in I$ and $q_k\in F$ the path is
successful. A word is {\it accepted} by a NFA if it's the label of a
successful path.  A NFA is {\it accessible} (resp.  {\it
co-accessible}) if for every state $q$ there exists a path from an
initial state to $q$ (resp. if for every state $q$ there exists a path
from $q$ to a final state).  A NFA is {\it trim} if it is both
accessible and co-accessible A {\it deterministic automaton} is a NFA
where $|I|=1$ and whose set of transitions is deterministic.

Let $\NFA(n)$ be the class of  finite automata whose set of states is
$\{0,\ldots,n-1\}$. We are interesting in several subclasses of $\NFA(n)$.
\begin{itemize}
\item $\N(n)$ is the subclass of $\NFA(n)$ of trim finite automata.
\item  $\N_m(n)$ be the class of
finite automata in  $\N(n)$ such that, for each state $p$, there is at most $m$
pairs $(a,q)$ such that $(p,a,q)$ is a transition: there are at most $m$
transitions outgoing each state. 
\item $\N^\prime_m(n)$ be the class of
finite automata in  $\N_m(n)$ such that, for each state $p$ and each
letter $a$, there is at most $m$
states $q$ such that $(p,a,q)$ is a transition: for each letter, there are
at most $m$ transitions labeled by this letter  outgoing each state.
\item Finally, for any class $\X$ of finite automata, we denote by $\X^{\bullet}$ the
subclass of $\X$ of automata whose set of initial states is reduced to
$\{1\}$.  
\end{itemize}

Examples of automata in One
has $$\N_m(n)\subseteq \N_m^\prime(n)\subseteq \N(n)\subseteq\NFA(n).$$
these classes are depicted on Fig.~\ref{fig:AutomatesClasses}.

\begin{figure}[!ht]
\centering
\begin{tikzpicture}
\node (1) [state,draw,initial, initial text=,fill=black!20] at (0,0) {$1$} ;
\node (2) [state,draw,fill=black!20] at (3,0) {$2$} ;
\node (3) [state,accepting,fill=black!20] at (6,0) {$3$} ;

\path[->,>=latex] (1) edge[above] node {$a$} (2);
\path[->,>=latex] (1) edge[above,bend left] node {$a$} (3);
\path[->,>=latex] (1) edge[loop above] node {$b$} (1);
\path[->,>=latex] (2) edge[above] node {$a,b$} (3);
\path[->,>=latex] (2) edge[below, bend left] node {$a,b$} (1);

\draw (3,-2) node {
\begin{tabular}{c}
Automaton  in $\N(3)^\bullet$, $\N_4(3)^\bullet$, $\N_2'(3)^\bullet$\\
and in   $\N(3)$, $\N_4(3)$, $\N_2'(3)$.
\end{tabular}};







\begin{scope}[xshift=0cm,yshift=-5cm]
\node (1) [state,draw,initial, initial text=,fill=black!20] at (0,0) {$1$} ;
\node (2) [state,draw,fill=black!20] at (3,0) {$2$} ;
\node (3) [state,accepting,fill=black!20,initial, initial
text=,fill=black!20] at (6,0) {$3$} ;

\path[->,>=latex] (1) edge[above] node {$a$} (2);
\path[->,>=latex] (1) edge[loop above] node {$a$} (1);
\path[->,>=latex] (2) edge[above,bend left] node {$b$} (3);
\path[->,>=latex] (2) edge[below, bend left] node {$b$} (1);

\draw (3,-2) node {
\begin{tabular}{c}
Automaton  in $\N(3)$, $\N_2(3)$ and $\N_2'(3)$\\
but not in any doted class.
\end{tabular} };
\end{scope}
\end{tikzpicture}
\caption{Several Classes of Automata.}\label{fig:AutomatesClasses}
\end{figure}


Two NFAs are {\it isomorphic} if there exists a 
bijection between their sets of
states preserving the sets initial states, final states and transitions. More precisely,
let $\A=(Q,\Sigma,\Delta,I,F)$ and let $\varphi$ be a bijection from $Q$
into a finite set $\varphi(Q)$. We denote by $\varphi(\A)$ the automaton
$(\varphi(Q),\Sigma,\Delta^\prime,\varphi(I),\varphi(F))$, with 
$\Delta^\prime=\{(\varphi(p),a,\varphi(q))\mid (p,a,q)\in \Delta\}$. 
Two automata $\A_1$ and $\A_2$ are isomorphic if there exists a bijection
$\varphi$ such that $\varphi(\A_1)=\A_2$. 

\begin{figure}
\begin{center}
\begin{tikzpicture}[scale=1.2]
\node [circle,draw,initial,initial text=,fill=black!20](1) at (-2.5,0) {$1$};
\node [circle,draw,fill=black!20](2) at (0,0) {$2$};
\node [circle,draw,fill=black!20](3) at (0,-2) {$3$};
\node [circle,draw,fill=black!20,accepting](4) at (-2.5,-2) {$4$};

\path[->,>=latex] (1) edge[above] node {$a$} (2);
\path[->,>=latex] (2) edge[bend left,left] node {$b$} (3);
\path[->,>=latex] (3) edge[bend left,left] node {$b$} (2);
\path[->,>=latex] (3) edge[left,above] node {$a$} (4);
\path[->,>=latex] (4) edge[loop below] node {$b$} (4);
\path[->,>=latex] (4) edge[left] node {$a$} (1);

\node [circle,draw,initial,initial text=,fill=black!20](1b) at (2.5,0) {$2$};
\node [circle,draw,fill=black!20](2b) at (5,-2) {$3$};
\node [circle,draw,fill=black!20](3b) at (2.5,-2) {$4$};
\node [circle,draw,fill=black!20,accepting](4b) at (5,-0) {$1$};

\path[->,>=latex] (1b) edge[right] node {$a$} (2b);
\path[->,>=latex] (2b) edge[bend left, below] node {$b$} (3b);
\path[->,>=latex] (3b) edge[bend left, below] node {$b$} (2b);
\path[->,>=latex] (3b) edge[left] node {$a$} (4b);
\path[->,>=latex] (4b) edge[loop below] node {$b$} (4b);
\path[->,>=latex] (4b) edge[above] node {$a$} (1b);

\path[->,>=latex,dashed] (1) edge[above,color=black,bend left] node {} (1b);
\path[->,>=latex,dashed] (2) edge[above,color=black] node {} (2b);
\path[->,>=latex,dashed] (3) edge[above,color=black,bend left] node {} (3b);
\path[->,>=latex,dashed] (4) edge[above,color=black] node {} (4b);

\draw (-1,-3.5) node {$\A_1$};
\draw (3.5,-3.5) node {$\A_2$};

\draw (1,-4.5) node {$\varphi(\A_1)=\A_2,\; \textcolor{black}{\varphi(1)=2,\, \varphi(2)=3,\, \varphi(3)=4,\, \varphi(4)=1}$.};
\end{tikzpicture}




\end{center}
\caption{Two Isomorphic Automata}\label{fig:isom}
\end{figure}

Two isomorphic NFAs have the same number of states and are equal, up to
 the states names. The relation {\it is isomorphic to} is an equivalence relation.
For instance, the two automata depicted on Fig.~\ref{fig:isom} are
isomorphic.
An {\it automorphism} for a NFA is an isomorphism between this NFA and
itself. 
Given a NFA $\A=(Q,\Sigma,\Delta,I,F)$, the set of automorphisms of $\A$ is
a finite group denoted $\Aut(\A)$. For $Q^\prime\subseteq Q$, 
$\Aut_{Q^\prime}(\A)$ denotes the subset of  $\Aut(\A)$ of automorphisms $\phi$
fixing each element of $Q^\prime$:  for each $q\in Q^\prime$, $\phi(q)=q$.
Particularly $\Aut_{\emptyset}(\A)=\Aut(\A)$, and  
$\Aut_{Q}(\A)$ is reduce to the identity.
For instance, the  automorphism group of the 
automaton depicted on Fig.~\ref{fig:isom}(a) has two elements, the identity
and the isomorphism switching $2$ and $3$.

The size of the automorphism group of a non deterministic $n$-state
automaton may vary from $1$ to $n!$. For instance, any deterministic trim
automaton whose states are all final has an  automorphism group reduce to
the identity. The non deterministic $n$-state automaton with no transition
and where all states are both initial and final has for automorphism group
the symmetric group. 

The isomorphism problem consists in deciding whether two finite automata are
isomorphic. It is investigated for deterministic automata on different alphabet
in~\cite{DBLP:journals/siamcomp/Booth78} (with a different definition of
isomorphism). It is naturally closed to the same problem for directed graph
and the following result~\cite{DBLP:journals/jcss/Luks82} will be useful in
this paper.
\begin{theorem}\label{thm:luks}
Let $m$ be a fixed positive integer. The isomorphism problem for directed
graphs with degree bounded by $m$ is polynomial.
\end{theorem}

\subsection{Theoretical Background on Markov Chains}\label{sec:theoryMC}

For a general reference on Markov Chains see~\cite{mixing}. Basic
probability notions will not be defined in this paper. The reader is
referred for instance to~\cite{proba-and-computing}.

Let $\Omega$ be a finite set. A {\it Markov chain} on $\Omega$ is a sequence
$X_0,\ldots,X_t,\ldots$ of random variables on $\Omega$ such that
$\P(X_{t+1}=x_{t+1}\mid X_t=x_t)=\P(X_{t+1}=x_{t+1}\mid
X_t=x_t,\ldots,X_i=x_i,\ldots,X_0=x_0),$ for all $x_i\in\Omega$. A Markov
chain is defined by its {\it transition matrix} $M$, which is a function
from $\Omega\times\Omega$ into $[0,1]$ satisfying $M(x,y)=\P(X_{t+1}=y \mid
X_t=x)$. The underlying graph of a Markov chain is the graph whose set of
vertices is $\Omega$ and there is an edge from $x$ to $y$ if $M(x,y)\neq 0$.
A Markov chain is {\it irreducible} if its underlying graph is strongly
connected. It is {\it aperiodic} if for all node $x$, the gcd of the lengths
of all cycles 
visiting $x$ is $1$. Particularly, if for each $x$, $M(x,x)\neq
0$, the Markov chain is aperiodic. A Markov chain is {\it ergodic} if it is
irreducible and aperiodic. A Markov chain is symmetric if $M(x,y)=M(y,x)$
for all $x,y\in\Omega$. A distribution $\pi$ on $\Omega$ is a stationary
distribution for the Markov Chain if $\pi M =\pi$. It is known that an
ergodic Markov chain has a unique stationary
distribution~\cite[Chapter~1]{mixing}. Moreover, if the chain is symmetric,
this distribution is the uniform distribution on $\Omega$.

Given an ergodic Markov chain $X_0,\ldots,X_t,\ldots$ with stationary
distribution $\pi$, it is known that, whatever is the value of $X_0$, the
distribution of $X_t$ converges to $\pi$ when $t\to +\infty$:
$\max \tv{M^t(x,\cdot)-\pi}\underset{t\to +\infty}{\to} 0,$ where $\tv{}$
designates the total variation distance between two
distributions~\cite[Chapter~4]{mixing}. This leads to the Monte-Carlo
technique to randomly generate elements of $\Omega$ according to the
distribution $\pi$ by choosing arbitrarily $X_0$, computing $X_1,X_2,\ldots$, and
returning $X_t$ for $t$ large enough. The convergence rate is known to be
exponential, but computing the constants is a very difficult problem: choosing
the step $t$ to stop is a challenging question depending both on how close to
$\pi$ we want to be and on the convergence rate of $M^t(x,\cdot)$ to $\pi$.
For this purpose, the $\varepsilon$-{\it mixing time} of an ergodic Markov
chain of matrix $M$ and stationary distribution $\pi$ is defined by $t_{\rm
mix}(\varepsilon)= \min \{t\mid \max_{x\in \Omega}\tv{P_t(x,\cdot)-\pi}\leq \varepsilon\}.$
Computing mixing time bounds is a central question on Markov Chains.

The Metropolis-Hasting Algorithm is based on the Monte-Carlo technique and
aims at modifying the transition matrix of the Markov chain in order to
obtain a particular stationary distribution~\cite[Chapter 3]{mixing}.
Suppose that $M$ is an ergodic symmetric transition matrix of a symmetric
Markov chain on $\Omega$ and $\nu$ is a distribution on $\Omega$. The
transition matrix $P_\nu$ for $\nu$ is defined by:

$$
P_\nu(x,y)=
\left\{
\begin{array}{ll}
\min\left\{1,\frac{\nu(y)}{\nu(x)}\right\} M(x,y)& \text{if } x\neq y,\\
1-\sum_{z\neq x}\min\left\{1,\frac{\nu(z)}{\nu(x)}\right\} M(x,z)& \text{if } x= y.
\end{array}\right. 
$$

The chain defined by $P_\nu$ is called {\it the Metropolis Chain for $\nu$}.
It is known~\cite[Chapter 3]{mixing} that it is an ergodic Markov chain
whose stationary distribution is~$\nu$.

\section{Random Generation of Non Deterministic Automata using Markov Chain}\label{sec:random}

In this section, we propose families of symmetric ergodic Markov
chains on $\NFA(n)$, $\N(n)$, $\N_m(n)$ and $\N_m^\prime(n)$, as well
as on the respective corresponding doted classes of NFAs. The movement of
theses Markov chains are depicted in Fig.~\ref{fig:NFAMarkov}.

\begin{figure}[!ht]
\begin{center}
\begin{tikzpicture}[scale=1.2]

\node [circle,minimum width=2.5cm,draw,fill=red!10](A1) at (0,0) {};
\node [rectangle,draw,fill=white] (1a) at (-0.6,-0.2){$\A$};
\node [circle,draw,fill=white,accepting] (2a) at (0.5,-0.3){{\tiny $q$}};
\node [circle,draw,fill=white] (3a) at (0,0.5){{\tiny $p$}};
\node (fl) at (-.1,-.8) {};

\path[->,>=latex]  (fl) edge[] (2a);
\path[->,>=latex] (1a) edge[left] node {{\tiny $a$}} (3a);
\path[->,>=latex] (1a) edge[below] node {{\tiny $b$}} (2a);
\path[->,>=latex] (3a) edge[left] node {{\tiny $a$}} (2a);

\draw[->,>=latex] (A1) to [out=300,in=270,looseness=15] (A1);

\begin{scope}[rotate=-50]
\node [circle,minimum width=2.5cm,draw,fill=blue!10](A2) at (4,0) {};
\node [rectangle,draw,fill=white] (1b) at (3.4,-0.2){$\A$};
\node [circle,draw,fill=white,accepting] (2b) at (4.5,-0.3){{\tiny $q$}};
\node [circle,draw,fill=white] (3b) at (4,0.5){{\tiny $p$}};

\node (flb) at (3.9,-.8) {};

\path[->,>=latex]  (flb) edge[] (2b);

\path[->,>=latex] (1b) edge[above] node {{\tiny $a$}} (3b);
\path[->,>=latex] (1b) edge[left] node {{\tiny $b$}} (2b);
\path[->,>=latex] (3b) edge[left] node {{\tiny $a$}} (2b);
\path[->,>=latex] (3b) edge[right,bend left,color=red] node {{\tiny $b$}} (2b);

\path[->,>=latex] (A1) edge[right] node {$\frac{\rho_3}{|\Sigma|.n^2}$} (A2);
\end{scope}
\begin{scope}[rotate=50]
\node [circle,minimum width=2.5cm,draw,fill=blue!10](A2) at (4,0) {};
\node [rectangle,draw,fill=white] (1b) at (3.4,-0.2){$\A$};
\node [circle,draw,fill=white,accepting] (2b) at (4.5,-0.3){{\tiny $q$}};
\node [circle,draw,fill=white] (3b) at (4,0.5){{\tiny $p$}};

\node (fp) at (3.3,0.7) {};
\path[->,>=latex,red]  (fp) edge[] (3b);

\node (flc) at (3.9,-.8) {};
\path[->,>=latex]  (flc) edge[] (2b);

\path[->,>=latex] (1b) edge[left] node {{\tiny $a$}} (3b);
\path[->,>=latex] (1b) edge[below] node {{\tiny $b$}} (2b);
\path[->,>=latex] (3b) edge[above] node {{\tiny $a$}} (2b);

\path[->,>=latex] (A1) edge[left] node {$\frac{\rho_1}{|Q|}$} (A2);
\end{scope}

\begin{scope}[rotate=110]
\node [circle,minimum width=2.5cm,draw,fill=blue!10](A2) at (4,0) {};
\node [rectangle,draw,fill=white] (1b) at (3.4,-0.2){$\A$};
\node [circle,draw,fill=white,accepting] (2b) at (4.5,-0.3){{\tiny $q$}};
\node [circle,draw,fill=white] (3b) at (4,0.5){{\tiny $p$}};

\path[->,>=latex] (1b) edge[below] node {{\tiny $a$}} (3b);
\path[->,>=latex] (1b) edge[right] node {{\tiny $b$}} (2b);
\path[->,>=latex] (3b) edge[left] node {{\tiny $a$}} (2b);

\path[->,>=latex] (A1) edge[left] node {$\frac{\rho_1}{|Q|}$} (A2);
\end{scope}

\begin{scope}[rotate=170]
\node [circle,minimum width=2.5cm,draw,fill=blue!10](A2) at (4,0) {};
\node [rectangle,draw,fill=white] (1b) at (3.4,-0.2){$\A$};
\node [circle,draw,fill=white,accepting] (2b) at (4.5,-0.3){{\tiny $q$}};
\node [circle,draw=red,fill=white,accepting] (3b) at (4,0.5){{\tiny $p$}};

\node (flc) at (3.9,-.8) {};
\path[->,>=latex]  (flc) edge[] (2b);

\path[->,>=latex] (1b) edge[right] node {{\tiny $a$}} (3b);
\path[->,>=latex] (1b) edge[above] node {{\tiny $b$}} (2b);
\path[->,>=latex] (3b) edge[left] node {{\tiny $a$}} (2b);

\path[->,>=latex] (A1) edge[below] node {$\frac{\rho_2}{|Q|}$} (A2);
\end{scope}

\begin{scope}[rotate=-130]
\node [circle,minimum width=2.5cm,draw,fill=blue!10](A2) at (4,0) {};
\node [rectangle,draw,fill=white] (1b) at (3.4,-0.2){$\A$};
\node [circle,draw,fill=white] (2b) at (4.5,-0.3){{\tiny $q$}};
\node [circle,draw,fill=white] (3b) at (4,0.5){{\tiny $p$}};

\node (flc) at (3.9,-.8) {};
\path[->,>=latex]  (flc) edge[] (2b);

\path[->,>=latex] (1b) edge[right] node {{\tiny $a$}} (3b);
\path[->,>=latex] (1b) edge[above] node {{\tiny $b$}} (2b);
\path[->,>=latex] (3b) edge[below] node {{\tiny $a$}} (2b);

\path[->,>=latex] (A1) edge[right] node {$\frac{\rho_2}{|Q|}$} (A2);
\end{scope}

\node [circle,minimum width=2.5cm,draw,fill=blue!10](A2) at (4,0) {};
\node [rectangle,draw,fill=white] (1b) at (3.4,-0.2){$\A$};
\node [circle,draw,fill=white,accepting] (2b) at (4.5,-0.3){{\tiny $q$}};
\node [circle,draw,fill=white] (3b) at (4,0.5){{\tiny $p$}};
\path[->,>=latex] (1b) edge[left] node {{\tiny $a$}} (3b);
\path[->,>=latex] (3b) edge[left] node {{\tiny $a$}} (2b);

\node (flc) at (3.9,-.8) {};
\path[->,>=latex]  (flc) edge[] (2b);

\path[->,>=latex] (A1) edge[above] node {$\frac{\rho_3}{|\Sigma|.n^2}$} (A2);

\begin{scope}[rotate=25]
\node (A3) at (5,0) {};
\path[->,>=latex,dashed] (A1) edge[] (A3);
\end{scope}
\begin{scope}[rotate=80]
\node (A3) at (5,0) {};
\path[->,>=latex,dashed] (A1) edge[] (A3);
\end{scope}
\begin{scope}[rotate=140]
\node (A3) at (5,0) {};
\path[->,>=latex,dashed] (A1) edge[] (A3);
\end{scope}
\begin{scope}[rotate=-105]
\node (A3) at (5,0) {};
\path[->,>=latex,dashed] (A1) edge[] (A3);
\end{scope}
\begin{scope}[rotate=205]
\node (A3) at (5,0) {};
\path[->,>=latex,dashed] (A1) edge[] (A3);
\end{scope}

\node [rectangle,draw,fill=white] (bbb) at (-3,-6){$\A$};
\draw (2,-6) node  {is used to describe the unmodified part of the automaton.};
\end{tikzpicture}
\end{center}
\caption{Moves into the Markov Chain. The current state of the Markov Chain
is in the center and, around, typical moves. Of course, other moves are
possible on other parts of the automaton, which is represented by the dashed arrows.\label{fig:NFAMarkov}}
\end{figure}

Let $\A=(Q,\Sigma,\Delta,I,F)$ be a finite automaton. 
For any $q$ in $Q$ and any $(p,a,q)$ in $Q\times\Sigma\times Q$, the
automata $\ChI(\A,q)$, $\ChF(\A,q)$ and $\ChT(\A,(p,a,q))$ are defined
as follows:
\begin{itemize}
\item 
If $q\in I$,  then $\ChI(\A,q)=(Q,\Sigma,\Delta,I\setminus\{q\},F)$ and
$\ChI(\A,q)=(Q,\Sigma,\Delta,I\cup\{q\},F)$ otherwise.

\item 
If $q\in F$, then $\ChF(\A,q)=(Q,\Sigma,\Delta,I,F\setminus\{q\})$, and
$\ChF(\A,q)=(Q,\Sigma,\Delta,I,F\cup\{q\})$ otherwise.

\item 
 If $(p,a,q)\in\Delta$, then
$\ChT(\A,(p,a,q))=(Q,\Sigma,\Delta\setminus\{(p,a,q)\},I,F)$, and
$\ChT(\A,(p,a,q))=(Q,\Sigma,\Delta\cup\{(p,a,q)\},I,F)$ otherwise.
\end{itemize}
 Let $\rho_1$, $\rho_2$, $\rho_3$ be three real numbers
satisfying $0\leq \rho_i\leq 1$ and $\rho_1+\rho_2+\rho_3\leq 1$. Let $\X$
be a class of automata whose set of states is $Q$. 
We define the transition matrix $\S^\X(x,y)$ on $\X$ by:
\begin{itemize}
\item If there exists $q$ such that 
$y=\ChI(x,q)$, then $\S^\X(x,y)=\frac{\rho_1}{|Q|}$.
\item If there exists $q$ such that 
$y=\ChF(x,q)$, then $\S^\X(x,y)=\frac{\rho_2}{|Q|}$.
\item If there exists $(p,a,q)\in Q\times\Sigma\times Q$ such that 
$y=\ChT(x,q)$, then $\S^\X(x,y)=\frac{\rho_3}{|\Sigma|.|Q|^2}$.
\item If $y$ is different of $x$ and has not one of the above forms, 
$\S^\X(x,y)=0$.
\item $\S^\X(x,x)=1-\sum_{y\neq x} \S^\X(x,y)$.
\end{itemize}

Now for $\X\in \{\N(n),\N_m(n),\N_m^\prime(n)\}$, and $0< \rho<1$ we define
 the transition matrix $S^{\X^{\bullet}}_\rho$ on $\X^{\bullet}$ by
 $S^{\X^{\bullet}}_\rho=S^\X_{0,\rho,1-\rho}$.

\begin{lemma}\label{lm:irreductible}
Let $m,n$ be fixed positive integers, with $m\geq 2$.
If $1>\rho >0$, $\rho_1>0$, $\rho_2 >0$ and $\rho_3>0$, then  $\S^{\N(n)}$,
$\S^{\N_m(n)}$ and $\S^{\N_m^\prime(n)}$ are irreducible, as well as   $\Sr^{\N(n)^{\bullet}}$,
$\Sr^{\N_m(n)^{\bullet}}$ and $\Sr^{\N_m^\prime(n)^{\bullet}}$. 
\end{lemma}

\begin{proof}
Without loss of generality, we assume that $Q=\{1,\ldots,n\}$. Let
$\X\in\{\N(n),\N_m(n),\N_m^\prime(n)\}$ and $x\in \X$. We denote by $\A_0$
the automaton $(Q,\Sigma,\emptyset,Q,Q)$. The automaton $\A_0$ is trim and
is in $\X$. We prove there is  a path in $\X$ from $x$ to $\A_0$. Set
$x=(Q,\Sigma,\Delta,I,F)$. Since adding initial or final states to $x$ provides
automata that are still in $\X$, there is a path from $x$ to
$y=(Q,\Sigma,\Delta,Q,Q)$ (using $\ChI$ and $\ChF$). Now, since all states are
both initial and final, there is a path from $y$ to $\A_0$ (by deleting all
transitions). It follows there is a path in $\X$ from $x$ to $\A_0$. Since
the graph of the Markov chain is symmetric, there is also a path from $\A_0$ to $x$.
Consequently, the Markov chains are irreducible. 

We will now consider the bullet classes. The proof is only done for
$\N_m(n)^{\bullet}$. Proof for others classes are similar. 
Let $a_0$ be an arbitrary letter of $\Sigma$.  Let
$x=(\{1,\ldots,n\},\Sigma,\Delta,\{1\},F)\in \N_m(n)^{\bullet}$. Let
$\A_1$ be the automaton of $\N_m(n)^{\bullet}$ whose set of final
states is $\{1,\ldots,n\}$ and whose set of transitions is
$\{(i,a_0,i+1)\mid 1\leq i < n\}$. We will prove that there is a path
from $x$ to $\A_1$  in the graph of the Markov chain. 

\begin{itemize}
\item By adding final states, there is a path from $x$ to the automaton 
  $y=(\{1,\ldots,n\},\Sigma,\Delta,\{1\},\{1,\ldots,n\})$.

\item Let $\Delta^\prime$ be a subset of $\Delta$ forming a spanning
  tree of $\Delta$ rooted in $1$ (it exists for $y$ is
  accessible). By removing transitions, there is a path from 
$y$ to $z=(\{1,\ldots,n\},\Sigma,\Delta^\prime,\{1\},\{1,\ldots,n\})$.

\item If $z$ there are in $z$ at least two states $p$ and $q$ that
  have no outgoing transition ($p$ and $q$ are leaves of the spanning
  tree), then by adding a transition $(p,a_0,q)$ and by removing the
  unique transition arriving in $q$, we build a path from $z$ to an
  automaton inducing a tree on $,\{1,\ldots,n\}$ and having strictly
  less leaves. By repeating this kind of moves, there is a path from
  $z$ to an automaton of the form 
 $$u=(\{1,\ldots,n\},\Sigma,\{(\varphi(i),a_i,\varphi(i+1))\mid 1\leq i
  <n\},\{1\},\{1,\ldots,n\}),$$ where $\varphi$ is a permutation of
  $\{1,\ldots,n\}$ fixing $1$.

\item Since $m\geq 2$, one can add a transition leaving each
  state. Therefore, by adding each transition of the form
  $(i,a_0,i+1)$. Next, one can remove all transitions that are not of the from
 $(i,a_0,i+1)$, providing a path from $u$ to $\A_1$. 
\end{itemize}

It follows that there is a path from $x$ to $\A_1$ and, since the
graph is symmetric, a path from $\A_1$ to $x$, proving that the Markov
chain is irreducible.~\endofproof
\end{proof}

\begin{lemma}\label{lm:aperiodic}
Let $m,n$ be two fixed positive integers, with $m\geq 2$.
If $1>\rho >0$, $\rho_1>0$, $\rho_2 >0$ and $\rho_3>0$, then $\S^{\N(n)}$,
$\S^{\N_m(n)}$ and $\S^{\N_m^\prime(n)}$ are aperiodic,  as well as   $\Sr^{\N(n)^{\bullet}}$,
$\Sr^{\N_m(n)^{\bullet}}$ and $\Sr^{\N_m^\prime(n)^{\bullet}}$. 
\end{lemma}

\begin{proof}
With the notations of the proof of Lemma~\ref{lm:irreductible}, there is a
path of length $n_x$ from any $x\in \X$ to $\A_0$. Therefore there is a 
cycle of length $2n_x$ visiting $x$. 

Now, $\ChI(\A_0,1)\notin \X$ since $1$ is not accessible in $\A_0$. It
follows that $S^\X(\A_0,\A_0)\neq 0$. Therefore, there is also a 
cycle of length $2n_x+1$ visiting $x$. Since the gcd of $2n_x$ and $2n_x +1$ is $1$,
the chain is aperiodic.

The proof for   $\Sr^{\N(n)^{\bullet}}$ is similar:
From $\A_1$, if the transition $(1,a_0,2)$ is picked up (with
probability $\frac{(1-\rho)}{|\Sigma|n^2}\neq 0$), then we move from
$\A_1$ to $\A_1$.  Therefore there is a loop of length $1$ from $\A_1$
to $\A_1$ in the graph of the Markov chain, proving that the Markov
chain is aperiodic.

One can do as well for
$\Sr^{\N_m(n)^{\bullet}}$ and $\Sr^{\N_m^\prime(n)^{\bullet}}$.
\endofproof
\end{proof}

\begin{proposition}\label{prop:ergodic}
Let $m,n$ be two fixed positive integers, with $m\geq 2$.
The Markov chains with matrix  $\S^{\N(n)}$,
$\S^{\N_m(n)}$ and $\S^{\N_m^\prime(n)}$ are ergodic and their stationary distributions
are the uniform distributions.
\end{proposition}

\begin{proof}
By lemma~\ref{lm:aperiodic} and~\ref{lm:irreductible}, the chain is ergodic.
Since the matrix $\S^{\N(n)}$, $\S^{\N_m(n)}$ and $\S^{\N_m^\prime(n)}$ are
symmetric, their stationary distributions are the uniform distributions
(over the respective family of automata).
\endofproof
\end{proof}

In practice, computing $X_{t+1}$ from $X_t$ is done in the following way:
the first step consists in choosing with probabilities $\rho_1$, $\rho_2$
and $\rho_3$ whether we will change either an initial state, a final state
or a transition. In a second step and in each case, all the possible
changing operations are performed with the same probability. If the obtained
automaton is in the corresponding class, $X_{t+1}$ is set to this value.
Otherwise, $X_{t+1}=X_t$. Since verifying that an automaton is in the
desired class ($\N(n)$, $\N_m(n)$ or $\N_m^\prime(n)$), can be performed in
time polynomial in $n$, computing $X_{t+1}$ from $X_t$ can be done in time
polynomial in $n$.



We define the lazy Markov chain on $\NFA(n)$ by
$L_{\rho_1,\rho_2,\rho_3}^{\NFA(n)}(x,y)=\frac{1}{2}\S^{\NFA(n)}(x,y)$ if $x\neq
y$ and
$L_{\rho_1,\rho_2,\rho_3}^{\NFA(n)}(x,x)=\frac{1}{2}+\frac{1}{2}\S^{\NFA(n)}(x,x)$.
It is known that a symmetric Markov chain and its associated lazy
Markov chain have similar mixing times.

\begin{proposition}\label{prop:mixing}
The $\varepsilon$-mixing time $\tau(\varepsilon)$ of
$L_{\rho_1,\rho_2,\rho_3}^{\NFA(n)}$ satisfies
\begin{align*}
\tau(\varepsilon)\leqslant \max &\left(\left\lceil\frac{n}{\rho_1}\left(\log(n)+\log\left(\frac{1}{\rho_1\varepsilon}\right)\right)\right\rceil,
 \left\lceil\frac{n}{\rho_2}\left(\log(n),\log\left(\frac{1}{\rho_2\varepsilon}\right)\right)\right\rceil,\right.\\
&  \left.\left\lceil\frac{|\Sigma|n^2}{\rho_3}\left(\log(|\Sigma|n^2)+\log\left(\frac{1}{\rho_3\varepsilon}\right)\right)\right\rceil\right) .
\end{align*}
\end{proposition}

\begin{proof}

The proof lies on classical results on random walks int he hypercube. 

To each finite automaton $\A=(Q,\Sigma,\Delta,I,F)$ in $\NFA(n)$ one can
associate a function  $\varphi_I$ from $Q=\{1,\ldots,n\}$ to  $\{0,1\}$ by:

$$
\forall q\in Q,\, \varphi_I(q)=\begin{cases}1 \text{ if } q\in I\\ 0 \text{ otherwise.}\end{cases}
$$

In a same way, the functions   ${\varphi _F: Q\to \{0,1\}}$ and
${\varphi_\Delta : Q\times\Sigma\times Q}$ are defined as the characteristic
function s $1_F$ and $1_\Delta$ on $F$ and $\Delta$
considered as subsets of respectively  $Q$ and $Q\times \Sigma\times Q$.

The automaton $\A$ is completely defined by  $\varphi_I$, $\varphi_F$ and
$\varphi_\Delta$ (since the alphabet is fixed). Therefore, the Markov chain
can be decomposed into three random walks~: the two first on hypercubes of
dimension $n$ and the last one on an hypercube of dimension  $n^2|\Sigma|$
(for the transitions). 
At each step, one moves with probability  $\rho_1$ in the first hypercube,
with probability  $\rho_2$ in the second hypercube and with probability
$\rho_3$ in the last hypercube. 

The result is then the application of known mixing time results for lazy
random walk in the
hypercube, see for instance~\cite[Section 6.5.2, page 81]{mixing}.~\endofproof
\end{proof}

At this stage, we are not able to compute bounds on the mixing times
of the other Markov chains.  Practical experiments, with various sizes
of alphabets, seems to show that most of the automata generated by
the above lazy Markov Chain (using $n^3$ as mixing bound) are trim. The
experimental results are reported in Table~\ref{tab:rejets}.
This observation leads us to consider, for other experiments, to move
$n^3$ steps to sample automata. Of course, this is not a proof, just
an  empirical estimation.

\begin{table}
\begin{center}
\begin{tabular}{|c|c|c|c|c|c|}
\hline
$n\to$ &  5 & 10& 15 &20  \\\hline
$\N(n)$, 2 letters & 0.94&  0.95  &  0.99 &  1.0     \\\hline
$\N(n)$, 4 letterers & 0.95 &  0.96  & 1.0  & 1.0     \\\hline
\end{tabular}
\end{center}

\caption{Observed proportion of trim automata.}\label{tab:rejets}
\end{table}

\section{Metropolis Hastings Approach}\label{sec:randomup}
In this section we show how to use the Metropolis-Hastings algorithm to
uniformly generate NFAs up to isomorphism and that, for this purpose, it
{\it suffices} to compute the sizes of the automorphism groups of involved
NFAs. We prove in Section~\ref{sec:counting} that this computation is
polynomially equivalent to testing the isomorphism problem for the involving
automata. For the classes $\N_m(n)$, $\N_m(n)^{\bullet}$, $\N_m^\prime(n)$
and $\N_m^\prime(n)^{\bullet}$, we show that it can be done in time
polynomial in $n$ (if $m$ is fixed). In Section~\ref{sec:labellings} we show
how to practically compute the sizes of automorphism group using labelings
techniques. Finally, experimental results are given in Section~\ref{sec:xp}.

\subsection{Metropolis-Hastings Algorithm}\label{sec:algo-metro}

For a class $\C$ of NFAs (closed by isomorphism) and $n$ a positive integer,  let
 $\C(n)$ be the elements of $\C$ whose set of states is $\{1,\ldots,n\}$ and
 let $\gamma_n$ be the number of isomorphism classes on $\C(n)$.
There are $n!$ possible bijections on $\{1,\ldots,n\}$. If $\A\in\C(n)$, 
Let $\varphi_1$ and $\varphi_2$ be two bijections on $\{1,\ldots,n\}$. One
has  $\varphi_1(\A)=\varphi_2(\A)$ iff $\varphi_2^{-1}\varphi_1(\A)=\A$, iff
 $\varphi_2^{-1}\varphi_1(\A)\in\Aut(\A)$. It follows that  the isomorphism
 classes of $\A$ (in $\C(n)$) has $\frac{n!}{|\Aut(\A)|}$ elements. 
This leads to the following result.

\begin{proposition}\label{prop:MH}
Randomly generates an element $x$ of $\C(n)$ with probability 
$\frac{|\Aut(x)|}{\gamma_nn!}$ provides a uniform random generator of the
isomorphism classes of  $\C(n)$.
\end{proposition}

\begin{proof}
Let $H$ be an isomorphism class of $\C(n)$; $H$ is generated with probability
$$\sum_{x\in H}\frac{|\Aut(x)|}{\gamma_nn!}=\sum_{x\in
  H}\frac{1}{\gamma_n|H|}=\frac{1}{\gamma_n|H|} \sum_{x\in
  H}1=\frac{|H|}{\gamma_n|H|}= \frac{1}{\gamma_n}.$$
\endofproof
\end{proof}

In order to compute $P_\nu$ it is not necessary to compute $\gamma_n$, since
$\frac{\nu(x)}{\nu(y)}=\frac{|\Aut(y)|}{|\Aut(x)|}$. A direct use of the
Metropolis-Hastings algorithm requires to compute all the neighbors of $x$
and the sizes of theirs automorphism groups to move from $x$. Since a
$n$-state automaton has about $|\Sigma|n^2$ neighbors, it can be a quite
huge computation for each move. However, practical evaluations show that in
most cases the automorphism group of an automaton is quite small and,
therefore, the rejection approach exposed in~\cite{Chib} is more tractable.
It consists in moving from $x$ to $y$ using $S(x,y)$ (the non-modified
chain) and to accept $y$ with probability
$\min\left\{1,\frac{\nu(y)}{\nu(x)}\right\}$.
If it is not accepted, repeat the process (moving from $x$ to $y$ using $S$
with probability $\min\left\{1,\frac{\nu(y)}{\nu(x)}\right\}$) until
acceptance. In practice, we observe a very small number of rejects. 

The problem of computing the size of the automorphism group of a NFA is
investigated in the next session. Assuming it can be done in a reasonable
time, an alternative solution to randomly generate NFAs up to isomorphism
may be to use a rejection algorithm: randomly and uniformly generate a NFA
$\A$ and keep it with probability  $\frac{|\Aut(\A)|}{n!}$. This way, each
class of isomorphism is picked up with the same probability. However, as we
will observe in the experiments (see Table~\ref{fig:nbisom1}), most of automata have a very small group of
automorphisms, and the number of rejects will be intractable, even for quite
small $n$'s.

\subsection{Counting Automorphisms}\label{sec:counting}

This section is dedicated to show how to compute $|\Aut(\A)|$ by using a
polynomial number of calls to the isomorphism problem. It is an adaptation
of a corresponding result for directed
graphs~\cite{DBLP:journals/ipl/Mathon79}.






Let $\A=(Q,\Sigma,\Delta,I,F)$ be a NFA and $Q^\prime\subseteq Q$. Let
$\sigma$ be an arbitrary bijective function from $Q^\prime$ into
$\{1,\ldots,|Q^\prime|\}$, $a_0$ an arbitrary letter in $\Sigma$ and $\ell=
|Q|+|Q^\prime|+2$. For each state $r\in Q\backslash Q^{\prime}$ we denote by
$\A_r^{Q^\prime}$ the automaton $(Q_r,\Sigma,\Delta_r,I,F)$ where
$Q_r=Q\cup\{(p,i)\mid p\in Q\text{ and } 1\leq i \leq \ell\}$, and
$\Delta_r=\Delta\cup\{(p,a,(p,1)\mid p\in Q\}\cup \{((p,i),a_0,(p,i+1))\mid
p\in Q^\prime\text{ and } 1\leq i <
|Q|+1+\sigma(p)\}\cup \{((r,i),a_0,(r,i+1))\mid 1\leq\
i \leq \ell \}\cup \{((p,i),a_0,(p,i+1))\mid p\notin Q^\prime\cup\{r\}\text{
and } 1< i \leq |Q|+1\}.$ Note that the size of $\A_r^{Q^\prime}$ is
polynomial in the size of $\A$. 

The two next lemma show how to polynomially reduce the problem of counting
automorphisms to the isomorphism problem. 

\begin{lemma}\label{lem:auto1}
Let $\A=(Q,\Sigma,\Delta,I,F)$ be a NFA  and $Q^\prime$ a
non-empty subset of $Q$. For every $q,q^\prime\in Q\backslash
Q^{\prime}$,  there
exists $\phi\in \Aut_{Q^\prime}(\A)$ such that $\phi(q)=q^\prime$ iff
$\A_q^{Q^\prime}$ and $\A_{q^\prime}^{Q^\prime}$ are isomorphic. 
\end{lemma}

\begin{proof}
In $\A_r$, for any state $p\in Q^\prime\cup\{r\}$, we denote by 
$\pi_p$ the path
$$\pi_p=(p,a_0,(p,1))((p,1),a_0,(p,2))\ldots ((p,\ell-1),a_0,(p,\ell)),$$
called the {\it tail} of $p$.

Assume  that there exists $\phi\in \Aut_{Q^\prime}(\A)$ such that
$\phi(q)=q^\prime$. Let $\hat{\phi}$ be the function defined from the set of
states of $\A_q$ into the set of $\A_{q^\prime}$ by: if $p\in Q$, then
$\hat{\phi}(p)=\phi(p)$ and if $(p,i)$ is a state of $\A_q$, then
$\hat{\phi}((p,i))=(\phi(p),i)$. This function is well defined since $p$ and
$\varphi(p)$ have tails of the same length. By construction, $\hat{\phi}$
is an isomorphism. 

Conversely, assume that there exists an isomorphism $\Phi$ from $\A_q$ to
$\A_{q^\prime}$. Since isomorphisms preserve accessible and co-accessible
states and since $\A$ is trim, $p\in Q$ iff $\Phi(p)\in Q$. Let $\phi$ be
the restriction of $\Phi$ to $Q$. Since $\Phi$ is a morphism, $\phi$ is an
automorphism of $\A$. Now, $\Phi$ preserves the lengths of the tails. It
follows that for any $p\in Q^\prime$, $\Phi(p)=p$. Furthermore,
$\Phi(q)=q^\prime$ for the same reason, proving the lemma.~\endofproof 
\end{proof}

\begin{lemma}\label{lem:auto2}
 Let $\A=(Q,\Sigma,\Delta,I,F)$ be a NFA  and $Q^\prime$ a
non-empty subset of $Q$. For every $q\in Q^\prime$, there exists an integer
$d$ such that $|\Aut_{Q^\prime\setminus\{q\}}(\A)|=d|\Aut_{Q^\prime}(\A)|$.
Moreover $d$ can be computed with a polynomial number of isomorphism tests
between automata of the form  $\A_{r}^{Q^\prime\setminus \{q\}}$. 
\end{lemma}

\begin{proof}
Let $d=|\{\phi(q)\mid \phi\in \Aut_{Q^\prime\setminus\{q\}}(\A)\}|.$ We
consider the relation $\sim_q$ on $\Aut_{Q^\prime\setminus\{q\}}(\A)$
defined by $\phi_1\sim_q\phi_2$ iff $\phi_1(q)=\phi_2(q)$. One has
$\phi_1\sim_q\phi_2$ iff $\phi_1\phi_2^{-1}\in \Aut_{Q^\prime}(\A)$. Therefore
$\sim$ is a group-congruence relation and, therefore, 
$|\Aut_{Q^\prime\setminus\{q\}}(\A)|=d|\Aut_{Q^\prime}(\A)|$.

To compute $d$ it suffices to test which elements of $Q\setminus Q^\prime$
are in $\{\phi(q)\mid \phi\in \Aut_{Q^\prime\setminus\{q\}}(\A)\}$,
whether there exists
an automorphism $\phi$ in $\Aut_{Q^\prime\setminus\{q\}}(\A)$ such that
$\phi(q)=p$.~\endofproof 
\end{proof}

 Lemma~\ref{lem:auto2} provides a  way to
compute sizes of automorphism groups by testing whether two NFAs are
isomorphic. Indeed, since $\Aut_Q(\A)$ is reduced to the  identity, 
and since $\Aut(\A)=\Aut_{\emptyset}(\A)$, one has, by a direct induction using
Lemma~\ref{lem:auto2}, $\Aut(\A)=d_1\ldots d_{|Q|}$, where each $d_i$ can be
computed by a polynomial number of isomorphism tests. Therefore, the problem
of counting automorphism reduces to test whether two automata are
isomorphic.

\section{Polynomial Time Result for Specific Classes}\label{sec:poly}

\subsection{Isomorphism Problem for Automata with a Bounded Degree}\label{sec:isom}

It is proved (not explicitly) in~\cite{DBLP:journals/siamcomp/Booth78} that
the isomorphism problem for deterministic automata (with a different
   notion of isomorphismsince automata may
  have different alphabets)  is polynomially equivalent
to the isomorphism problem for directed finite graphs. We prove
(Theorem~\ref{thm-isom}) a similar result for NFAs, by using an encoding
preserving some bounds on the output degree. Therefore, combining
Theorem~\ref{thm-isom} and Lemma~\ref{lem:auto2}, it is possible to
compute the size of the automorphism group of an automaton in $\N_m(n)$,
$\N^\prime_m(n)$, $\N_m(n)^{\bullet}$ and $\N^\prime_m(n)^{\bullet}$ in time
polynomial in $n$ (assuming that $m$ is a constant).

\begin{theorem}\label{thm-isom}
Let $m$ be a fixed integer. The isomorphism problem for automata in $\N_m$
$\N^\prime_m$, $\N_m(n)^{\bullet}$ and $\N^\prime_m(n)^{\bullet}$
can be solved in polynomial time.
\end{theorem}

The proof of the theorem is given in appendix and is based on a graph
encoding of non deterministic automata and on
Theorem~\ref{thm:luks}~\cite{DBLP:journals/jcss/Luks82}. Note that the proof
is constructive but the exponents are too huge to provide an efficient
algorithm. It will be possible to work on a finer encoding but we prefer, in
practice, to use labeling techniques described in the next section and that
are practically very efficient on graphs (see~\cite{gallian} for a recent
survey).
\subsection{Proof of Theorem~\ref{thm-isom}}


Let $h$ be an arbitrary bijective function from $\Sigma$ into
$\{1,\ldots,k\}$.

Let $\A=(Q,\Sigma,\Delta,I,F)$ be a finite automaton. We denote by
$G_\A$ the finite graph $(V,E)$ where:

\begin{itemize}
\item $V=Q\cup (I\cap F^c)\times\{1,\ldots,k+1\}\cup (F\cap
  I^c)\times\{1,\ldots,k+2\}\cup (I\cap F)\times\{1,\ldots,k+3\}\cup
  \{((p,a,q),i)\mid (p,a,q)\in\Delta\text{ and } 1\leq i \leq
  h(a)\}\cup\Delta.$

\item $E=\{(p,(p,a,q))\mid (p,a,q)\in \Delta\}\cup 
\{((p,a,q)),q)\mid (p,a,q)\in \Delta\} \cup
\{((p,a,q),((p,a,q),1))\mid (p,a,q)\in \Delta\}\cup
\{(((p,a,q),i),((p,a,q),j))\mid (p,a,q)\in \Delta \text{ and } 
1\leq i,j \leq h(a)\}\cup
\{((q,i),(q,j))\mid q\in I\cap F^c\text{ and }1\leq i,j \leq k+1 \}\cup
\{((q,i),(q,j))\mid q\in F\cap I^c\text{ and }1\leq i,j \leq k+2 \}\cup
\{((q,i),(q,j))\mid q\in F\cap I\text{ and }1\leq i,j \leq k+3 \}\cup
\{(q,(q,1))\mid q\in I\cup F)\}
$
\end{itemize}

Intuitively each transition is first decomposed into two edges, then a complete graph
(with a number of edges depending on the letter of the transition) 
is linked to the middle vertex. A similar construction is done for initial,
final or both initial and final states.

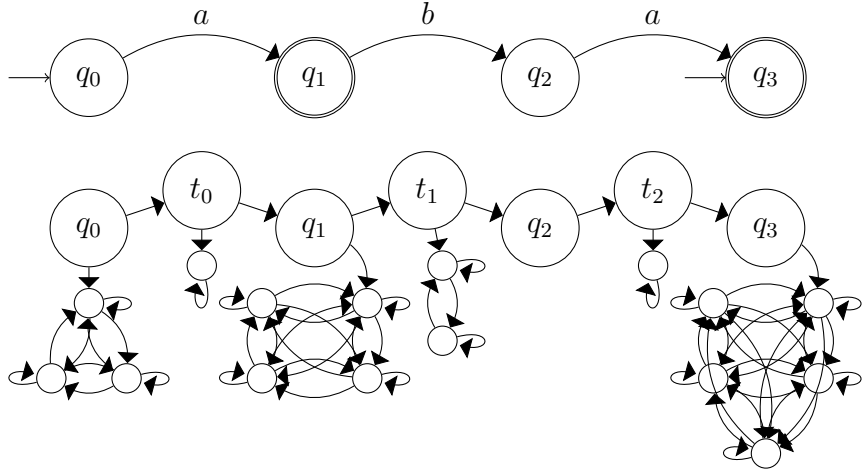
\begin{figure}
\begin{center}
\begin{tikzpicture}
\node[draw,state,initial, initial text=] (0) at (0,0) {$q_0$};
\node[draw,state, accepting] (1) at (3,0) {$q_1$};
\node[draw,state] (2) at (6,0) {$q_2$};
\node[draw,state,accepting, initial, initial text=] (3) at (9,0) {$q_3$};

\path[->,>=triangle 90] (0) [] edge[above,bend left] node {$a$} (1);
\path[->,>=triangle 90] (1) [] edge[above,bend left] node {$b$} (2);
\path[->,>=triangle 90] (2) [] edge[above,bend left] node {$a$} (3);

\node[draw,state] (0) at (0,-2) {$q_0$};
\node[draw,circle] (q01) at (0,-3){};
\node[draw,circle] (q02) at (-.5,-4){};
\node[draw,circle] (q03) at (.5,-4){};

\path[->,>=triangle 90] (q01) [] edge[above,bend left] node {} (q02);
\path[->,>=triangle 90] (q01) [] edge[above,bend left] node {} (q03);
\path[->,>=triangle 90] (q02) [] edge[above,bend left] node {} (q01);
\path[->,>=triangle 90] (q02) [] edge[above,bend left] node {} (q03);
\path[->,>=triangle 90] (q03) [] edge[above,bend left] node {} (q01);
\path[->,>=triangle 90] (q03) [] edge[above,bend left] node {} (q02);
\path[->,>=triangle 90] (q03) [] edge[above,loop right] node {} ();
\path[->,>=triangle 90] (q01) [] edge[above,loop right] node {} ();
\path[->,>=triangle 90] (q02) [] edge[above,loop left] node {} ();
\path[->,>=triangle 90] (0) [] edge[above] node {} (q01);

\node[draw,state] (1) at (3,-2) {$q_1$};
\node[draw,circle] (q11) at (3.7,-3){};
\node[draw,circle] (q12) at (2.3,-4){};
\node[draw,circle] (q13) at (3.7,-4){};
\node[draw,circle] (q14) at (2.3,-3){};
\path[->,>=triangle 90] (q11) [] edge[above,bend left] node {} (q12);
\path[->,>=triangle 90] (q11) [] edge[above,bend left] node {} (q13);
\path[->,>=triangle 90] (q11) [] edge[above,bend left] node {} (q14);
\path[->,>=triangle 90] (q12) [] edge[above,bend left] node {} (q11);
\path[->,>=triangle 90] (q12) [] edge[above,bend left] node {} (q13);
\path[->,>=triangle 90] (q12) [] edge[above,bend left] node {} (q14);
\path[->,>=triangle 90] (q13) [] edge[above,bend left] node {} (q11);
\path[->,>=triangle 90] (q13) [] edge[above,bend left] node {} (q12);
\path[->,>=triangle 90] (q13) [] edge[above,bend left] node {} (q14);
\path[->,>=triangle 90] (q14) [] edge[above,bend left] node {} (q11);
\path[->,>=triangle 90] (q14) [] edge[above,bend left] node {} (q12);
\path[->,>=triangle 90] (q14) [] edge[above,bend left] node {} (q13);
\path[->,>=triangle 90] (q13) [] edge[above,loop right] node {} ();
\path[->,>=triangle 90] (q11) [] edge[above,loop right] node {} ();
\path[->,>=triangle 90] (q12) [] edge[above,loop left] node {} ();
\path[->,>=triangle 90] (q14) [] edge[above,loop left] node {} ();
\path[->,>=triangle 90] (1) [] edge[above, bend left] node {} (q11);

\node[draw,state] (2) at (6,-2) {$q_2$};

\node[draw,state] (3) at (9,-2) {$q_3$};
\node[draw,circle] (q31) at (9.7,-3){};
\node[draw,circle] (q32) at (8.3,-4){};
\node[draw,circle] (q33) at (9.7,-4){};
\node[draw,circle] (q34) at (8.3,-3){};
\node[draw,circle] (q35) at (9,-5){};
\path[->,>=triangle 90] (q35) [] edge[above,bend left] node {} (q31);
\path[->,>=triangle 90] (q35) [] edge[above,bend left] node {} (q32);
\path[->,>=triangle 90] (q35) [] edge[above,bend left] node {} (q33);
\path[->,>=triangle 90] (q35) [] edge[above,bend left] node {} (q34);
\path[->,>=triangle 90] (q31) [] edge[above,bend left] node {} (q32);
\path[->,>=triangle 90] (q31) [] edge[above,bend left] node {} (q35);
\path[->,>=triangle 90] (q31) [] edge[above,bend left] node {} (q33);
\path[->,>=triangle 90] (q31) [] edge[above,bend left] node {} (q34);
\path[->,>=triangle 90] (q32) [] edge[above,bend left] node {} (q31);
\path[->,>=triangle 90] (q32) [] edge[above,bend left] node {} (q33);
\path[->,>=triangle 90] (q32) [] edge[above,bend left] node {} (q35);
\path[->,>=triangle 90] (q32) [] edge[above,bend left] node {} (q34);
\path[->,>=triangle 90] (q33) [] edge[above,bend left] node {} (q31);
\path[->,>=triangle 90] (q33) [] edge[above,bend left] node {} (q32);
\path[->,>=triangle 90] (q33) [] edge[above,bend left] node {} (q35);
\path[->,>=triangle 90] (q33) [] edge[above,bend left] node {} (q34);
\path[->,>=triangle 90] (q34) [] edge[above,bend left] node {} (q31);
\path[->,>=triangle 90] (q34) [] edge[above,bend left] node {} (q32);
\path[->,>=triangle 90] (q34) [] edge[above,bend left] node {} (q35);
\path[->,>=triangle 90] (q34) [] edge[above,bend left] node {} (q33);
\path[->,>=triangle 90] (q33) [] edge[above,loop right] node {} ();
\path[->,>=triangle 90] (q31) [] edge[above,loop right] node {} ();
\path[->,>=triangle 90] (q32) [] edge[above,loop left] node {} ();
\path[->,>=triangle 90] (q34) [] edge[above,loop left] node {} ();
\path[->,>=triangle 90] (q35) [] edge[above,loop left] node {} ();
\path[->,>=triangle 90] (3) [] edge[above, bend left] node {} (q31);

\node[draw,state] (0-1) at (1.5,-1.5) {$t_0$};
\node[draw,circle] (t01) at (1.5,-2.5){};
\path[->,>=triangle 90] (t01) [] edge[loop below] node {} (t01);
\path[->,>=triangle 90] (0-1) [] edge[above] node {} (t01);

\node[draw,state] (1-2) at (4.5,-1.5) {$t_1$};
\node[draw,circle] (t11) at (4.7,-2.5){};
\node[draw,circle] (t12) at (4.7,-3.5){};
\path[->,>=triangle 90] (t11) [] edge[loop right] node {} (t11);
\path[->,>=triangle 90] (t12) [] edge[loop right] node {} (t12);
\path[->,>=triangle 90] (1-2) [] edge[above] node {} (t11);
\path[->,>=triangle 90] (t12) [] edge[above,bend left] node {} (t11);
\path[->,>=triangle 90] (t11) [] edge[above, bend left] node {} (t12);

\node[draw,state] (2-3) at (7.5,-1.5) {$t_2$};
\node[draw,circle] (t21) at (7.5,-2.5){};
\path[->,>=triangle 90] (t21) [] edge[loop below] node {} (t21);
\path[->,>=triangle 90] (2-3) [] edge[above] node {} (t21);

\path[->,>=triangle 90] (0) [] edge[above] node {} (0-1);
\path[->,>=triangle 90] (0-1) [] edge[above] node {} (1);
\path[->,>=triangle 90] (1) [] edge[above] node {} (1-2);
\path[->,>=triangle 90] (1-2) [] edge[above] node {} (2);
\path[->,>=triangle 90] (2) [] edge[above] node {} (2-3);
\path[->,>=triangle 90] (2-3) [] edge[above] node {} (3);
\end{tikzpicture}
\end{center}
\caption{Example $\A$ and $G_\A$, $h(a)=1$ and $h(b)=2$}
\end{figure}

For any vertex $s$ of $G_\A=(V,E)$, we denote by $d(s)$ the largest possible size of a clique in $G_\A$ containing $s$. More formally, one has
$$d(s)=\max\{|H|,\; H\times H\subseteq E \text{ and } s\in H\}|$$ 
Note that we may have $d(s)=0$.

\begin{lemma}
For any vertex $s$ of $G_\A$, one has $0\leq d(s)\leq k+3$.
\end{lemma}

\begin{proof}
By construction.
\end{proof}

\begin{lemma}\label{lemma:isom1}
Let $\A$ be a finite automaton. If there is in $G_\A$ an edge of the form
$(s,t)$ then,
\begin{enumerate}
\item If $d(s)=0$, then $s$ is in $Q\cup\Delta$,
\item If $0<d(s)\leq k$, then $s$ is of the form $((p,\ell,q),i)$, with $(p,\ell,q)\in\Delta$
  and $1\leq i \leq h^{-1}(\ell)$,
\item If $d(s)=k+1$, then  $s$ is of the form $(q,i)$ with 
$q\in I\cap F^c$,
\item If $d(s)=k+2$, then  $s$ is of the form $(q,i)$ with 
$q\in F\cap I^c$,
\item If $d(s)=k+3$, then  $s$ is of the form $(q,i)$ with 
$q\in F\cap I$.
\end{enumerate}
\end{lemma}

\begin{proof}
By construction.~\endofproof 
\end{proof}

\begin{lemma}\label{lemma:isom2}
Let $\A$ be a finite automaton. If there is in $G_\A$ an edge of the form
$(s,t)$ with $s\in Q$, then
\begin{enumerate}
\item $d(t)\in \{0,k+1,k+2,k+3\}$ and,
\item If $d(t)=0$, then $t$ is in $\Delta$,
\item If $d(t)=k+1$, then $s\in I\cap F^c$,
\item If $d(t)=k+2$, then $s\in F\cap I^c$,
\item If $d(t)=k+3$, then $s\in F\cap I$.
\end{enumerate}
\end{lemma}

\begin{proposition}\label{prop:autotograph}
Two finite automata $\A_1$ and $\A_2$ are isomorphic iff $G_{\A_1}$ and
$G_{\A_2}$ are isomorphic too.
\end{proposition}

\begin{proof}
By construction if $\A_1$ and $\A_2$ are isomorphic, then  $G_{\A_1}$ and
$G_{\A_2}$ are isomorphic too.

Now let $\A_1=(Q_1,A,\Delta_1,I_1,F_1)$ and $\A_2=(Q_2,A,\Delta_2,I_2,F_2)$
be two NFAs such that $G_{\A_1}$ and $G_{\A_2}$ are isomorphic. One can note
that for every vertex $s$ of $G_{\A_1}$, $d(\varphi(s))=d(s)$. It follows
that, $d(s)=0$ iff $d(\varphi(s))=0$. Therefore, by Lemma~\ref{lemma:isom1},
$\varphi$ induces a bijective map from $Q_1$ to $Q_2$. In the following we prove that  
the restriction of $\varphi$ to $Q_1$ 
is an isomorphism from $Q_1$ to $Q_2$.

\begin{itemize}
\item If $q\in I_1\cap F_1^c$, then $d((q,1))=k+1$. Therefore
  $d(\varphi((q,1)))=k+1$. Since $(q,(q,1))$ is an edge of $G_{\A_1}$,
  $(\varphi(q),\varphi((q,1)))$ is an edge of $G_{\A_2}$. Using
  the Assertion 3. of Lemma~\ref{lemma:isom2}, $\varphi(q)\in I_2\cap F_2^c$.
\item If $q\in I_1^c\cap F_1$, then $d((q,1))=k+2$. Therefore
  $d(\varphi((q,1)))=k+2$. Since $(q,(q,1))$ is an edge of $G_{\A_1}$,
  $(\varphi(q),\varphi((q,1)))$ is an edge of $G_{\A_2}$. Using the
  Assertion 4. of Lemma~\ref{lemma:isom2}, $\varphi(q)\in I_2\cap F_2^c$.
\item If $q\in I_1\cap F_1$, then $d((q,1))=k+3$. Therefore
  $d(\varphi((q,1)))=k+3$. Since $(q,(q,1))$ is an edge of $G_{\A_1}$,
  $(\varphi(q),\varphi((q,1)))$ is an edge of $G_{\A_2}$. Using the
  Assertion 5. of Lemma~\ref{lemma:isom2}, $\varphi(q)\in I_2\cap F_2^c$.
\item If $(p,a,q)\in \Delta_1$, then $d((p,a,q))=0)$. 
 Consequently  $d(\varphi((p,a,q)))=0)$.
 Since $(p,(p,a,q))$ is an edge in $G_{\A_1}$, 
 $(\varphi(p),\varphi((p,a,q)))$ is an edge in $G_{\A_2}$.
By Assertion 2. of Lemma~\ref{lemma:isom2}, $\varphi((p,a,q))\in\Delta_2$.
Set $\varphi((p,a,q))=(s,b,t)$, with $s,t\in Q_2$. The only ongoing
edge in $(s,b,t)$ in $G_{\A_2}$ is $(s,(s,b,t))$. Since $\varphi$ is an
isomorphism, $(\varphi^{-1}(s),(p,a,q))$ is an edge of $G_{\A_1}$. It
follows that $\varphi^{-1}(s)=p$. There are two outgoing edges from 
$(s,b,t)$ in $G_{\A_2}$: $((s,b,t),(1,(s,b,t))$ and  $((s,b,t),t)$.
The two outgoing edges from $(p,a,q)$ in $G_{\A_1}$ are 
$((p,a,q),1,(p,a,q))$ and $((p,a,q),q)$. 
 Since $d(q)=0$, $d(t)=0$, $d((1,(p,a,q)))=h(a)$ and $d(((s,b,t),t))=h(b)$,
 one necessarily has $\varphi(q)=t$ and $h(a)=h(b)$ (and therefore $a=b$).
In conclusion, we proved that if  $(p,a,q)\in \Delta_1$, then 
$(\varphi(p),a,\varphi(q))\in \Delta_2$. The same proof can be made using
$\varphi^{-1}$, proving the proposition.
\end{itemize}
\endofproof 
\end{proof}

\begin{proposition}\label{prop:convertion}
The size of $G_\A$ is polynomial in the size of $\A$. Moreover, if there is
at most $m$ outgoing transitions from a state of $\A$, the  degree of $G_\A$ is
bounded by $\max\{m+1,k+3\}$.
\end{proposition}

\begin{proof}
By construction, the size of $G_\A$ is polynomial in the size of $\A$.
Let $s$ be a vertex of $G_\A$. The following cases arise:
\begin{itemize}
\item If $s\in Q\cap I^c \cap F^c$, then the edges in $G_A$ starting from
  $s$ are all of the form $(s,t)$ where $t\in\Delta$ starts from $s$ (as a transition in $\A$).\\
  Therefore there are at most $m$ outgoing edges from $s$.

\item  If $s\in  I\cup F$, then 
 the edges in $G_A$ starting from $s$ are
  those of the form   $(s,t)$ where $t\in\Delta$ starts from $s$ (as a transition in $\A$), and the one from $s$ to $(s,(s,1))$.\\ Therefore there are at most $m+1$ outgoing
  edges from $s$. 

\item If $s\in \Delta$, then there are two outgoing edges from $s$.

\item If $s=(q,i)$, with $q\in Q$, then there are at most $k+3$ outgoing
  edges from $(q,i)$.
\item If $s=(t,i)$, with $t\in \Delta$, then there are at most $k$ outgoing
  edges from $(t,i)$, depending on the letter labeling $t$, proving the
  Proposition. 
\end{itemize}
\endofproof 
\end{proof}

Theorem~\ref{thm-isom} is a consequence of
Proposition~\ref{prop:convertion}, Proposition~\ref{prop:autotograph}
and Theorem~\ref{thm:luks}.

\section{Practical Approach using Labelings}\label{sec:label}

\begin{figure}[!ht]
\begin{center}
\begin{tikzpicture}[scale=1.2]
\node [] (A) at (0,0) {$(Q_1,\Sigma,\Delta_1,I_1,F_1)$}; 
\node [] (B) at (6,0) {$(Q_2,\Sigma,\Delta_2,I_2,F_2)$}; 
\path[->,>=latex] (A) edge[above] node {isom $\varphi$?} (B);

\node[draw, ellipse, minimum width=5cm,minimum height=5cm] (IFc1) at (0,-3){};
\node (Q1) at (-1.6,-2.2) {$Q_1$};
\node (Q2) at (7.6,-2.2) {$Q_2$};

\node[draw, ellipse, minimum width=5cm,minimum height=5cm] (IFc1) at (6,-3){};

\node[draw, ellipse, minimum width=1cm,fill=black!20] (IFc1) at (0,-2) {$I_1\cap F_1^c$};
\node[draw, ellipse, minimum width=1cm,fill=black!20] (IFc2) at (6,-2) {$I_2\cap F_2^c$};

\node[draw, ellipse, minimum width=1cm,fill=black!20] (IF1) at (-1,-3) {$I_1\cap F_1$};
\node[draw, ellipse, minimum width=1cm,fill=black!20] (IF2) at (5,-3) {$I_2\cap F_2$};

\node[draw, ellipse, minimum width=1cm,fill=black!20] (IcF1) at (1,-3) {$I_1^c\cap F_1$};
\node[draw, ellipse, minimum width=1cm,fill=black!20] (IcF2) at (7,-3) {$I_2^c\cap F_2$};

\node[draw, ellipse, minimum width=1cm,fill=black!20] (IcFc1) at (0,-4) {$I_1^c\cap F_1^c$};
\node[draw, ellipse, minimum width=1cm,fill=black!20] (IcFc2) at (6,-4) {$I_2^c\cap F_2^c$};

\path[->,>=latex] (IFc1) edge[above,color=black] node {$\varphi_1$?} (IFc2);
\path[->,>=latex] (IcFc1) edge[below,color=black] node {$\varphi_4$?} (IcFc2);
\path[->,>=latex] (IcF1) edge[above,color=black,out=-20,in=-160,pos=0.3] node {$\varphi_3$?} (IcF2);
\path[->,>=latex] (IF1) edge[below,color=black,out=20,in=160,pos=0.8] node {$\varphi_2$?} (IF2);
\end{tikzpicture}
\end{center}
\caption{Labelings Technique\label{fig:etiquettes}}
\end{figure}
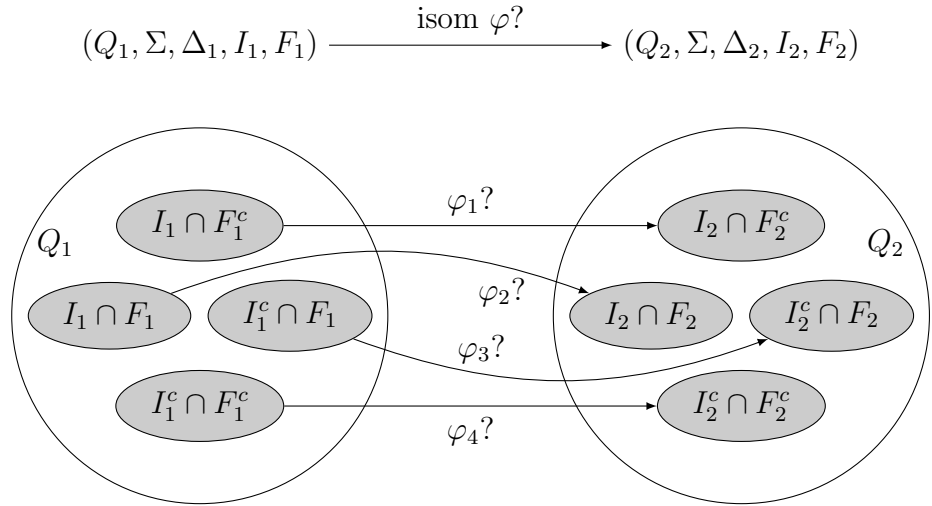

\subsection{Practical Computation using Labelings}\label{sec:labellings}

For testing graph isomorphism (or to
count the number of automorphisms), the most efficient currently used approach is
based on labeling~\cite{gallian} and it works practically for large graphs.
It can be naturally adapted for NFAs. Intuitively, to illustrate the
approach, one can note that if two $n$-state automata are isomorphic, then
they have the same number of initial states and of final states. Rather than
testing potential $n!$ possible bijections from the automata to point out an
isomorphism, it suffices to test $n_1! n_2! n_3! n_4!$ where $n_1$ is the
number of states that are both initial and final, $n_2$ the number of final
states (that are not initial), $n_3$ the number of initial states (that are
not final), and $n_4$ is the number of states that are neither initial, nor
final. With an optimal distribution, the number of tests falls from $n!$ to
$[(n/4)!]^4$. Of course, if all states are both initial and final, there are
still $n!$ bijection to test. This idea is illustrated in Fig.~\ref{fig:etiquettes}.

\begin{figure}[!ht]
{\bf Input:} $\A$, a $n$-state automaton, $\tau$ a labelling with image
$D=\{\alpha_1,\ldots,\alpha_\ell\}$. \\
{\bf Output:} $|\Aut(\A)|$
\smallskip

$res=0$\\
{\bf For} $\alpha\in D$\\
\sp  $C[\alpha]=\emptyset$\\
{\bf EndFor}\\
{\bf For} $i\in \{1,\ldots,n\}$\\
\sp $C[\tau(\A,i)]=C[\tau(\A,i)]\cup\{i\}$\\
{\bf EndFor}\\
{\bf Foreach} permutation $\sigma_1$ of $\tau^{-1}(\alpha_1)$\\
\sp {\bf Foreach} permutation $\sigma_2$ of $\tau^{-1}(\alpha_2)$\\
\sp\sp $\vdots$\\
\sp\sp {\bf Foreach} permutation $\sigma_\ell$ of $\tau^{-1}(\alpha_\ell)$\\
\sp\sp\sp {\bf If} $\sigma=\sigma_1\ldots\sigma_\ell\in \Aut(\A)$, {\bf
  Then}\\
\sp\sp\sp\sp $res=res+1$\\
\sp\sp\sp {\bf EndIf} \\
\sp\sp {\bf EndForeach}\\
\sp\sp $\vdots$\\
\sp {\bf EndForeach}\\
{\bf EndForeach}\\
{\bf Return} $res$
\caption{Counting automorphisms using labelings}\label{algo:count}
\end{figure}

This idea can be generalized by the notion of labeling; the goal is
to point out easily computable criteria that are stable by isomorphism to
get a partition of the set of states and to reduce the search. The approach can
be directly adapted for finite automata. A {\it labeling} is a computable
function $\tau$ from $\N(n)\times \{1,\ldots,n\}$ into a finite set $D$,
such that for $\A_1=(Q,\Sigma,E_1,I_1,F_1)$ and
$\A_2=(Q,\Sigma,E_2,I_2,F_2)$, if $\varphi$ is an isomorphism from $\A_1$ to
$\A_2$, then, for every $i\in \{1,\ldots,n\}$,
$\tau(\A_1,i)=\tau(\A_2,\varphi(i))$. The algorithm consists in looking for
functions $\varphi$ preserving $\tau$. If there exists $\alpha\in D$ such
that $|\{i \mid \tau(\A_1,i)=\alpha\}|\neq
|\{i \mid \tau(\A_2,i)=\alpha\}|$, then the two automata are not isomorphic.
Otherwise, all possible bijections preserving the labeling are tested. In
the worst case, there are $n!$ possibilities (the labeling doesn't provide
any refinement), but in practice, it works very well. Note that if $\tau_1$
and $\tau_2$ are two  labelings, then $\tau=(\tau_1,\tau_2)$ is a labeling to,
allowing the combination of labeling. In our work, we use the following
Labelings: the labeling testing whether a state is initial, the one
testing whether a state is final, the one testing whether a state is both
initial and final, the one returning, for each letter $a$, the number of
outgoing transitions labeled by $a$, the similar one with ongoing
transitions, the one returning the minimal word (in the lexical order) from
the state to a final state and the one returning the minimal word (in the
lexical order) from an initial state to the given state. 

Note that if the set of states is portioned into $p$ classes with $n/p$
elements en each classes, then one has to performed 
$t_n=((n/p)!)^p$ rahter than $n!$. A direct application of Stirling formula
show that 
$$\frac{n!}{\left((n/p)!\right)^p}\sim \sqrt{\frac{p^p}{(2n\pi)^{p-1}}} \cdot p^{n},$$
pointing out a significant theoretic complexity improvement. 

Using these Labelings the practical computation of the sizes of automorphism
groups can be done quite efficiently, using the algorithm depicted in
Fig.~\ref{algo:count}: first the set of states of the automaton is
partitioned into several subclasses according $\tau$. Since an automorphism
has to preserve these classes, one only explore this kind of automorphisms.
Note that if $D$ is large, the algorithm can be easily adapted to work on
$\alpha$'s such that $C[\alpha]\neq 0$.

We have computed sizes of automorphism using Markov chains. Computation is
very fast: labelings approaches allows to provides partitions of states
into subsets which are mostly singletons. Table~\ref{fig:nbisom1} reports
the results which are obtain, for each line, in less than a second. One can
also notice that most of generated automata have a small (compared to $n!$)
automorphism group.  

\begin{table}[!ht]
\begin{center}
\begin{tabular}{|l|c|c|c|}
\hline Class & av. size & maximal size\\
\hline
$\N_2(5)$ & 1.023 & 6\\ \hline
$\N_2(8)$ & 1.012 & 6\\ \hline
$\N_2(10)$& 1.015 & 2\\\hline
$\N_2(15)$& 1.007 & 2\\ \hline
$\N_2(20)$& 1.001 & 2\\ \hline
$\N_3(5$) & 1.031 & 6\\ \hline
$\N_3(8)$ & 1.015 & 2\\ \hline
$\N_3(10)$& 1.015 & 2\\ \hline
$\N_3(15)$& 1.005 & 2 \\ \hline
$\N_3(20)$& 1 & 1\\ \hline
$\N(5)$   & 1.022 & 2\\ \hline
$\N(8)$   & 1.01 & 2\\\hline
$\N(10)$  & 1.018 & 2\\ \hline
$\N(15)$  & 1.005  &2\\ \hline
$\N(20)$  & 1.002  &2\\ \hline

\end{tabular}
\end{center}
\caption{Sizes of automorphisms group, using a $n^3$ mixing time, 1000 tests
  for each line, $|\Sigma|=2$.}\label{fig:nbisom1}

\end{table}

\subsection{Experiments}\label{sec:xp}

 The experiments have
been done on a personal computer with processor {\tt IntelCore i3-4150
CPU 3.50GHz x 4}, {\tt 7,7 Go} of memory and running on a 64 bits Ubuntu
14.04 OS. The implementation is a non optimized prototype written in Python.

The first experimentation consists in measuring the time required to move into
the Metropolis chains for $\N(n)$ and $\N_m(m)$. Results are reported in
Table~\ref{tab:1}. The labelings used are those described in
Section~\ref{sec:labellings}. These preliminary results show that using a
$2$ or $3$-letter alphabet does not seem to have a significant influence. For
each generation, the $n^3$-th elements of the walk is returned, with an
arbitrary start. Moreover, bounding or not the degree does not seem to be
relevant for the computation time. 
Note that we do not use any optimization: several computations on
labelings may be reused when moving into the chain. Moreover, Python is not an
efficient programming language (compared to C or Java). In practice, for
directed graphs, the isomorphism problem is tractable for large graphs (see
for instance~\cite{DBLP:journals/ivc/FoggiaPSV09}). Note that the number of
moves ($n^3$) is the major factor for the increasing computation time
(relatively to $n$): the average time for moving a single step is multiplied by about
(only) $10$ from $n=20$ to $n=90$.

\begin{table}[!ht]
\begin{center}
\begin{tabular}{|c||c|c|c|c|c|}
\hline
$n$ & 10 & 20 & 50 & 70 & 90 \\
\hline
$|A|=2$ & 0.02 & 0.43& 32.5& 166.1& 569.9\\
$|A|=3$ & 0.02 & 0.56& 47.1& 248.4& 848.1\\
\hline 
\end{tabular}
\smallskip

\begin{tabular}{|c||c|c|c|c|c|}
\hline
$n$ & 10 & 20 & 50 & 70 & 90 \\
\hline
$m=2,|A|=2$ & 0.2 & 0.43& 32.5& 166.1& 566.8\\
$m=2,|A|=3$ & 0.2 & 0.57& 47.0& 246.7& 847.2\\
\hline 
$m=3,|A|=2$ & 0.2 & 0.43& 33.0& 167.8& 561.9\\
$m=3,|A|=3$ & 0.2 & 0.57& 47.2& 248.6& 851.3\\
\hline 
\end{tabular}
\end{center}
\smallskip
\caption{Average Time (s) to Sample a NFA in $\N(n)$ (left) and in
$\N_m^\prime(n)$ (right).}\label{tab:1}
\end{table}


For the last experience, we propose to compare our generation for
 $\N_2^\prime(n)^{\bullet}$ with the generator proposed
in~\cite{DBLP:conf/lpar/TabakovV05} with a density of $a$-transitions of $2$ and
$3$. The parameter of the algorithm is a probability $p_f$ for final states
and a density $\sigma$ on $a$-transitions: the set of states of the automaton is
$\{1,\ldots,n\}$, only $1$ is the initial state, each state is final with a
probability $p_f$ and for each $p$ and each $a$, $(p,a,q)$ is a transition with a
probability $\frac{\sigma}{n}$. Therefore for each state and each letter,
the expected number of outgoing transitions labeled by this letter is
$\sigma$. We run this algorithm with $p_f=0.2$ and $\sigma\in\{2,3\}.$ For
each size,
we compute the average size $s$ of the corresponding minimal automata. We use a two
letter alphabet and the average sizes (number of states) are obtained by
sampling 1000 automata for each case. Results are reported in
Table~\ref{tab:2}.


\begin{table}
\begin{center}
\begin{tabular}{|c||c|c|c|c|c|c|}
\hline
$\sigma=2$, $n=$ & 5 & 8 & 11 & 14 & 17 & 20 \\\hline
$s$ & 1.3 & 3.0 & 4.8 & 5.1 & 4.5 & 4.0\\
\hline
$\sigma=3$, $n=$ & 5 & 8 & 11 & 14 & 17 & 20 \\\hline
$s$ & 2.8 & 4.8 &4.7 & 3.8 & 3.4 & 3.0\\\hline
\end{tabular}

\smallskip
\begin{tabular}{|c||c|c|c|c|c|c|}
\hline
$\N_2^\prime(n)^{\bullet}$, $n=$ & 5 & 8 & 11 & 14 & 17 & 20 \\\hline
$s$ & 3.7 & 6.1 & 7.9 & 10.0 & 11.5 & 13.9\\
\hline
\end{tabular}
\end{center}
\caption{Average sizes of deterministic and minimal automata corresponding
to automata sampling using~\cite{DBLP:conf/lpar/TabakovV05} and in
$\N_2^\prime(n)^{\bullet}$.}\label{tab:2}
\end{table}


One can observe that the generator provides quite different automata. With
the Markov chain approach the sizes of the
related minimal automata are greater, even if there is no blow-up in
both cases. 

Finally, also to compare the proposed generator with~\cite{DBLP:conf/lpar/TabakovV05},
we have computed the sizes of the automorphisms groups. Results are
reported in Table~\ref{fig:nbisom1} and in Table~\ref{tab:nbisomvardi} and
points out significant differences.

\begin{table}
\begin{center}
\begin{tabular}{|c|c|c|c|c|c|c|}
\hline
& $p_i=0.2$& $p_i=0.5$ & $p_i=0.8$& $p_i=0.2$ &
$p_i=0.5$ & $p_i=0.8$ \\
& $p_f=0.2$& $p_f=0.5$ & $p_f=0.8$& $p_f=0.2$ &
$p_f=0.5$ & $p_f=0.8$ \\
& $\sigma=2$& $\sigma=2$& $\sigma=2$ & $\sigma=3$ & $\sigma=3$ &
$\sigma=3$\\\hline
n= 5 & 26 & 6 & 25 & 21&6&24\\\hline
n=8 & 2345 & 112 & 2213 & 2254 & 102& 2441\\\hline
n=10& 86 500 & 1343 & 71472 & 83072& 1303& 79203\\\hline
\end{tabular}
\end{center}
\caption{Observed average sizes of automata generated
  by~\cite{DBLP:conf/lpar/TabakovV05}, obtained with 1000 tests on a two
  letter alphabet.}\label{tab:nbisomvardi}
\end{table}

\section{Conclusion}

In this paper we proposed a Markov Chain approach to randomly generate non
deterministic automata (up to isomorphism) for several classes of NFAs. We
showed that moving into these Markov chains can be done quite quickly in
practice and, in some interesting cases, in polynomial time. Experiments
have been performed within a non optimized prototype and, following known
experimental results on group isomorphism, they allow us to think that the
approach can be used on much larger automata. Implementing such techniques
using an efficient programming language is a challenging perspective.
Moreover, the proposed approach is very flexible and can be applied to
various classes of NFAs. An interesting research direction is to design
particular subclasses of NFAs that look like NFAs occurring in practical
applications, even if this last notion is hard to define. We think that the
classes $\N_m^\prime(n)^{\bullet}$ and $\N_m(n)^{\bullet}$ constitute first
attempts in this direction. Theoretically -as often for Monte-Carlo
approach-, computing mixing and strong stationary times are crucial and
difficult questions we plan to investigate more deeply.



\bibliographystyle{elsarticle-num} 

\bibliography{biblio-nondet}

\end{document}